\documentclass[journal,twoside,web]{ieeecolor}
\usepackage{generic}
\usepackage{cite}
\usepackage{amsmath,amssymb,amsfonts}
\usepackage{empheq}
\usepackage{algorithm}
\usepackage{graphicx}
\usepackage{textcomp}
\usepackage{algorithmicx}
\usepackage{algpascal}

\DeclareMathAlphabet{\pazocal}{OMS}{zplm}{m}{n}
\newcommand{\norm}[1]{\left\lVert#1\right\rVert}
\newcommand{\Newmodel}{Signal Matrix Model}
\newcommand{\newmodel}{signal matrix model}
\newcommand{\newacro}{SMM}
\newtheorem{theorem}{Theorem}
\newtheorem{remark}{Remark}
\newtheorem{prop}{Proposition}
\newtheorem{defi}{Definition}

\usepackage{etoolbox}
\makeatletter
\@ifundefined{color@begingroup}%
  {\let\color@begingroup\relax
   \let\color@endgroup\relax}{}%
\def\fix@ieeecolor@hbox#1{%
  \hbox{\color@begingroup#1\color@endgroup}}
\patchcmd\@makecaption{\hbox}{\fix@ieeecolor@hbox}{}{\FAILED}
\patchcmd\@makecaption{\hbox}{\fix@ieeecolor@hbox}{}{\FAILED}

\def\BibTeX{{\rm B\kern-.05em{\sc i\kern-.025em b}\kern-.08em
    T\kern-.1667em\lower.7ex\hbox{E}\kern-.125emX}}
\begin{document}
\title{Maximum Likelihood Estimation in Data-Driven Modeling and Control}
\author{Mingzhou Yin, Andrea Iannelli, and Roy S. Smith, \IEEEmembership{Fellow, IEEE}% 
\thanks{This work was supported by the Swiss National Science Foundation under Grant 200021\_178890.}
\thanks{This work has been accepted for publication in the IEEE Transactions on Automatic Control.}
\thanks{© 2021 IEEE. Personal use of this material is permitted. Permission from IEEE must be obtained for all other uses, in any current or future media, including reprinting/republishing this material for advertising or promotional purposes, creating new collective works, for resale or redistribution to servers or lists, or reuse of any copyrighted component of this work in other works.}
\thanks{The authors are with the Automatic Control Laboratory, Swiss Federal Institute of Technology (ETH Z\"urich), 8092 Zurich, Switzerland (e-mail: myin@control.ee.ethz.ch; iannelli@control.ee.ethz.ch; rsmith@control.ee.ethz.ch).}}

\maketitle
\pagestyle{empty}
\thispagestyle{empty}

\begin{abstract}
Recently, various algorithms for data-driven simulation and control have been proposed based on the Willems' fundamental lemma. However, when collected data are noisy, these methods lead to ill-conditioned data-driven model structures. In this work, we present a maximum likelihood framework to obtain an optimal data-driven model, the \newmodel, in the presence of output noise. Data compression and noise level estimation schemes are also proposed to apply the algorithm efficiently to large datasets and unknown noise level scenarios. Two approaches in system identification and receding horizon control are developed based on the derived optimal estimator. The first one identifies a finite impulse response model% in combination with the kernel-based method
. This approach improves the least-squares estimator with less restrictive assumptions. The second one applies the \newmodel\ as the predictor in predictive control. The control performance is shown to be better than existing data-driven predictive control algorithms, especially under high noise levels. Both approaches demonstrate that the derived estimator provides a promising framework to apply data-driven algorithms to noisy data.
\end{abstract}


\section{Introduction}
\label{sec:introduction}
\IEEEPARstart{F}{ollowing} its remarkable success in artificial intelligence, learning from data is becoming a popular topic in various engineering domains \cite{Hou_2013}. This concept is by no means a new idea for control engineering. The system identification paradigm has been widely used in control applications, where data are used to fit an a priori parametrized model \cite{LjungBook2}. The control strategy is then designed with the identified nominal model based on the certainty equivalence principle.%, or with the uncertainty model based on robust/stochastic control frameworks, possibly with adaptive elements also learned from online data \cite{goodwin2014adaptive}.

However, this conventional scheme of learning dynamical systems is challenged by increasing complexity of systems and the large amount of data available. In particular, a low-dimensional model structure that is suitable to design compact, closed-form control strategies can be very hard and costly to obtain for complex systems \cite{Hjalmarsson_2005}. %In fact, low dimensionality is not really required in modern optimization-based control frameworks, and may limit the predictive power of big data \cite{Sutton_2019}.
Therefore, alternative paths are investigated to facilitate control design directly from raw measurement data of dynamical systems. %Early attempts in this direction include unfalsified control \cite{Safonov_1997}, iterative feedback tuning \cite{Hjalmarsson_1998}, and virtual reference feedback tuning \cite{Campi_2002}.
For example, reinforcement learning techniques are widely applied in this area \cite{Recht_2019}%, including policy search \cite{lagoudakis2003least} and approximate dynamic programming \cite{powell2007approximate}
. Such approaches typically avoid predicting the behavior of systems explicitly but aim at the control strategy directly.

In this work, the conventional parametric model is replaced by a data-driven predictor with a non-parametric structure \cite{Van_Waarde_2020}. In the seminal work from Willems \textit{et al. }\cite{Willems_2005}, a single input-output trajectory of the linear system is shown to be able to characterize all possible trajectories of length up to the order of persistency of excitation by constructing Hankel matrices from data. This result is known as the Willems’ fundamental lemma. With this result, the behavior of the system can be simulated and thus controlled by selecting a suitable combination of sections from the known trajectory that satisfies the initial condition constraints \cite{Markovsky_2008, DePersis_2020, vanWaarde_2020}.

This observation is especially suitable for optimal trajectory tracking. In this regard, model predictive control (MPC) is known to be very effective when an accurate model of the system is available \cite{camacho2016model}. From the Willems’ fundamental lemma, the output prediction step in MPC can be achieved by using known trajectories of the system directly, instead of an explicit model. This data-driven alternative to MPC algorithms, known as data-enabled predictive control (DeePC) \cite{Coulson_2019}, has lead to multiple successful applications \cite{Huang_2019, Coulson_2019_reg, Alpago_2020} with stability and robustness proofs \cite{Berberich_2020}. This framework is also able to handle online data and parameter variations \cite{pmlr-v144-yin21a}.

These types of ``data-driven" approaches differ significantly from model-free approaches, and act as a surrogate for conventional models in model-based control to provide a description of system trajectories based on measured data. In fact, with a low-rank approximation, this approach directly leads to the intersection algorithm in subspace identification where state-space models can be derived \cite{Moonen_1989}. The main differences of the data-driven approach compared to conventional model-based methods are: 1) the model is implicit with no closed-form solution in general; 2) the model is over-parametrized in that it does not impose any assumption on the system structure other than linearity. In this paper, this implicit and over-parametrized model is called the data-driven model.

However, it is well-known that when data are noisy, over-parametrized models may lead to high variances and overfitting \cite{Geman_1992}. In data-driven modeling, finding a combination of known trajectory sections that give reliable prediction is an ill-conditioned problem for datasets with stochastic noise. In current data-driven control schemes, empirical regularizers \cite{Berberich_2020,Coulson_2019} or least-norm problems \cite{Favoreel_1999,Huang_2019,Sedghizadeh_2018} are introduced to select a reasonable combination for prediction. Yet, it is not clear what is the optimal way to combine a large set of known trajectory sections to achieve the most reliable prediction. %In addition, the hyperparameters in the empirical regularizers are difficult to tune. %It has been observed that in DeePC the control performance is very sensitive to the hyperparameters in the regularizer \cite{Huang_2019}. The performance is usually assessed with an oracle of the optimal regularization parameters, which is not realistic in practical applications.

Another application of data-driven modeling is to simulate the system response \cite{Markovsky_2005b,Quintana_Carapia_2020}. The main advantage of applying this approach in system identification is that it gives the correct estimation of nonparametric models in the noise-free case. %This is a much-desired property yet fails to be satisfied by most existing non-parametric methods including least-squares regression \cite{Chen_2012} and empirical transfer function estimation \cite{LjungBook2}.
Again in this scenario, the best practice for solving the underdetermined linear system in the Willems' fundamental lemma in the noisy case is not understood. For computational simplicity, the Moore-Penrose pseudoinverse solution that solves the least-norm problem is often the default choice \cite{Huang_2019}, leading to the data-driven subspace predictor \cite{Sedghizadeh_2018}.

As can be seen from the above discussion, one of the central questions in data-driven approaches based on the Willems' fundamental lemma is how to obtain an optimal data-driven model from a large noise-corrupted dataset \cite{van2020noisy}. Therefore, in the first part of the paper, we propose a maximum likelihood estimation (MLE) framework to estimate such an optimal model with noise in both offline data and online measurements. This optimal model is named the \newmodel\ (\newacro). This framework optimizes the combination of offline trajectories by maximizing the conditional probability of observing the predicted output trajectory and the measured past outputs. The \newacro\ is shown to obtain more accurate output estimates than the least-norm solution. %Interestingly, it is shown that for Gaussian noise, the maximum likelihood estimator can be approximated by iteratively solving an $l_2$-norm minimization problem that is similar to the formulation of the regularization terms in the DeePC algorithm, but with adaptive hyperparameters.
In addition, a preconditioning strategy is proposed based on singular value decomposition (SVD) to compress the data matrix such that online complexity is %only dependent on the trajectory length to be simulated, to the benefit of 
fixed for large datasets. When the noise levels are unknown, they can also be estimated with a data-driven approach.

In the second part of the paper, we present two scenarios where the \newacro\ leads to effective algorithms: 1) estimating finite impulse response (FIR) models %with the kernel-based method
in system identification; and 2) obtaining a tuning-free data-driven predictive control scheme. In the first scenario, %the kernel-based method is used to derive a maximum a posteriori (MAP) estimator based on the impulse response estimate learned from data. Carefully designed kernels encode appropriate characteristics of the dynamical system.
the impulse response is conventionally estimated by least-squares regression, which requires knowing the input history and neglecting truncation errors. In this work, it is replaced by the \newmodel\ simulated with an impulse, which guarantees an unbiased estimate. Results show that the model fitting is enhanced when the transient response is unknown or the truncation error of the impulse response is large.

In the second scenario, we replace the prediction part in the DeePC algorithm with the \newacro. This predictor is shown to be superior to the pseudoinverse subspace predictor in predictive control. The main advantage of the proposed algorithm is that it avoids the difficult hyperparameter tuning problem in regularized DeePC. The control performance of the proposed algorithm is shown to be better than the DeePC algorithm with optimal hyperparameters when the noise is significant, and similar in the low noise scenario.

The remainder of the paper is organized as follows. Section~\ref{sec:2} defines the notions and preliminaries used in the paper. Section~\ref{sec:3} reviews the Willems' fundamental lemma and its application to deterministic systems. Section~\ref{sec:4} derives the \newmodel\ with MLE and presents an optimal data-driven simulation algorithm. Section~\ref{sec:4.5} discusses the use of \newacro\ for large datasets and unknown noise levels and analyzes its performance. This model is then applied to two problems: Section~\ref{sec:5} identifies an FIR model using \newacro\ simulation% and the kernel-based regularization
; Section~\ref{sec:6} applies the \newacro\ predictor in predictive control. Section~\ref{sec:7} concludes the paper.

\section{Notation \& Preliminaries}
\label{sec:2}

For a vector $x$, the weighted $l_2$-norm $(x^\mathsf{T}Px)^{\frac{1}{2}}$ is denoted by $\norm{x}_P$. The symbol $\pazocal{N}(\mu,\Sigma)$ indicates a Gaussian distribution with mean $\mu$ and covariance $\Sigma$. The expectation and the covariance of a random vector $x$ are denoted by $\mathbb{E}(x)$ and $\text{cov}(x)$ respectively. For a matrix $X$, the vectorization operator stacks its columns in a single vector and is denoted by $\text{vec}(X)$; $X^\dagger$ indicates the Moore-Penrose pseudoinverse; $(X)_{i,j}$ denotes the $(i,j)$-th entry of $X$. The symbol $\mathbb{S}_{++}^{n}$ indicates the set of $n$-by-$n$ positive definite matrix. For a sequence of matrices $X_1,\dots,X_n$, we denote $[X_1^\mathsf{T}\ \dots\ X_n^\mathsf{T}]^\mathsf{T}$ by $\text{col}\left(X_1,\dots,X_n\right)$. Given a signal $x:\mathbb{Z}\to \mathbb{R}^n$, its trajectory from $k$ to $k+N-1$ is denoted as $(x_i)_{i=k}^{k+N-1}$, and in the vector form as $\mathbf{x}=\text{col}(x_k,\dots,x_{k+N-1})$.

Consider a discrete-time linear time-invariant (LTI) system with output noise, given by
\begin{equation}
\begin{cases}
x_{t+1}&=\ A x_t+B u_t,\\
\hfil y_t&=\ C x_t + D u_t + w_t,
\label{eq:sys}
\end{cases}
\end{equation}
where $x_t \in \mathbb{R}^{n_x}$, $u_t \in \mathbb{R}^{n_u}$, $y_t \in \mathbb{R}^{n_y}$, $w_t \in \mathbb{R}^{n_y}$ are the states, inputs, outputs, and output noise respectively. The system is denoted compactly by $(A,B,C,D)$. The pair $(A,B)$ is controllable if $[B\ AB\ \dots\ A^{n_x-1}B]$ has full row rank.%; the pair $(A,C)$ is observable if $\text{col}\left(C,CA,\dots,CA^{n_x-1}\right)$ has full column rank. The system is minimal if $(A,B)$ is controllable and $(A,C)$ is observable.

The notion of persistency of excitation is defined as follows.
\begin{defi}
A signal trajectory $(x_i)_{i=0}^{N-1}\in \mathbb{R}^n\times\{0,\dots,N-1\}$ is said to be persistently exciting of order $L$ if the block Hankel matrix
\begin{equation}
    X=\begin{bmatrix}
    x_0&x_1&\cdots&x_{M-1}\\
    x_1&x_2&\cdots&x_{M}\\
    \vdots&\vdots&\ddots&\vdots\\
    x_{L-1}&x_L&\cdots&x_{N-1}
    \end{bmatrix}\in \mathbb{R}^{Ln\times M}
\end{equation}
\label{def:1}
\end{defi}
has full row rank, where $M=N-L+1$ \cite{Willems_2005}.

Intuitively, this definition means that sections of length $L$ of the trajectory span $\mathbb{R}^{Ln}$. When used as the input to a linear dynamical system, it can thus excite all the controllable behaviors of the system in a window of length $L$. A necessary condition of Definition~\ref{def:1} is $N\geq L(n+1)-1$, which gives a lower bound on the trajectory length.

\section{Deterministic Data-Driven Modeling}
\label{sec:3}

In this section, we first review the Willems' fundamental lemma and a few related results in a state-space formulation, followed by an overview of deterministic data-driven simulation and control.

\subsection{Willems' Fundamental Lemma}

Built on the notion of the persistency of excitation, the Willems' fundamental lemma shows that all the behavior of a linear system can be captured by a single persistently exciting trajectory of the system when no noise is present. This lemma was originally proposed in the context of behavioral system theory \cite{Willems_2005,willems1997introduction}, where systems are characterized by the subspace that contains all possible trajectories. It was later reformulated in the state-space context \cite{DePersis_2020,vanWaarde_2020}. In the state-space formulation, the output trajectory is unique to a particular input trajectory when a sufficiently long past input-output trajectory is specified as the initial condition. The length of the past trajectory should not be shorter than the state dimension. This idea has strong ties with the intersection algorithm in subspace identification \cite{Moonen_1989}, where a low-order subspace of the data matrices that corresponds to a low state dimension is sought.

We summarize the available results on data-driven modeling based on the Willems' fundamental lemma for finite-dimensional LTI systems, which are the foundation for the data-driven methods discussed in this paper. These results hold exactly only when the system is noise-free, i.e., $\forall i, w_i=0$.

\begin{theorem}
Consider a finite-dimensional LTI system $(A,B,C,D)$. Let $(u_i^d, x_i^d, y_i^d)_{i=0}^{N-1}$ be a triple of input-state-output trajectory of the system. If the pair $(A,B)$ is controllable and the input is persistently exciting of order $(L+n_x)$, then
\begin{itemize}
    \item[(a)]the matrix
    \begin{equation}
    \begin{bmatrix}X\\U\end{bmatrix}:=\begin{bmatrix}
        x_0^d&x_1^d&\cdots&x_{M-1}^d\\\hline
        u_0^d&u_1^d&\cdots&u_{M-1}^d\\
        \vdots&\vdots&\ddots&\vdots\\
        u_{L-1}^d&u_L^d&\cdots&u_{N-1}^d\\
        \end{bmatrix}
    \end{equation} has full row rank (Corollary 2 in \cite{Willems_2005}, Theorem 1(i) in \cite{vanWaarde_2020}, Lemma 1 in \cite{DePersis_2020});
    \item[(b)] the pair $(u_i,y_i)_{i=0}^{L-1}$ is an input-output trajectory of the system iff there exists $g$, such that
    \begin{equation}
    \begin{bmatrix}u_0\\\vdots\\u_{L-1}\\\hline    y_0\\\vdots\\y_{L-1}\end{bmatrix}=\begin{bmatrix}U\\Y\end{bmatrix}g:=\begin{bmatrix}
        u_0^d&u_1^d&\cdots&u_{M-1}^d\\
        \vdots&\vdots&\ddots&\vdots\\
        u_{L-1}^d&u_L^d&\cdots&u_{N-1}^d\\\hline
        y_0^d&y_1^d&\cdots&y_{M-1}^d\\
        \vdots&\vdots&\ddots&\vdots\\
        y_{L-1}^d&y_L^d&\cdots&y_{N-1}^d\\
        \end{bmatrix}g
    \end{equation}
    (Theorem 1 in \cite{Willems_2005}, Theorem 1(ii) in \cite{vanWaarde_2020}, Lemma 2 in \cite{DePersis_2020});
    \item[(c)] $\text{rank}(\text{col}(U,Y))=n_x+n_u L$ (Theorem 2 in \cite{Moonen_1989});
    %\item[(d)] if the system is minimal, the immediate past input-output trajectory $(u_i,y_i)_{i=-L_0}^{-1}$ uniquely determines the initial condition $x_0$, if $n_y L_0\geq n_x$ (Lemma 1 in \cite{Markovsky_2008});
    \item[(d)] the vector $(y_i)_{i=0}^{L'-1}$ is the unique output trajectory of the system with past trajectory $(u_i,y_i)_{i=-L_0}^{-1}$ and input trajectory $(u_i)_{i=0}^{L'-1}$, where $n_y L_0\geq n_x$ and $L'=L-L_0$, iff there exists $g$, such that
    \begin{equation}
    \begin{bmatrix}u_{-L_0}\\\vdots\\u_{L'-1}\\\hline    y_{-L_0}\\\vdots\\y_{L'-1}\end{bmatrix}=\begin{bmatrix}U\\Y\end{bmatrix}g:=\begin{bmatrix}
        u_0^d&u_1^d&\cdots&u_{M-1}^d\\
        \vdots&\vdots&\ddots&\vdots\\
        u_{L-1}^d&u_L^d&\cdots&u_{N-1}^d\\\hline
        y_0^d&y_1^d&\cdots&y_{M-1}^d\\
        \vdots&\vdots&\ddots&\vdots\\
        y_{L-1}^d&y_L^d&\cdots&y_{N-1}^d\\
        \end{bmatrix}g
        \label{eqn:fund}
    \end{equation}
    (Proposition 1 in \cite{Markovsky_2008}).
\end{itemize}
\label{thm:1}
\end{theorem}

\begin{remark}
    The controllability and the persistency of excitation conditions can be relaxed for the rank condition in part (c) (Corollary 19 in \cite{IM-FD}) or requirements on the initial state (Theorem 1 in \cite{yu2021controllability}).
\end{remark}

In Theorem \ref{thm:1}, parts (a) and (b) state the original Willems' fundamental lemma; part (c) draws the connection between data-driven modeling and subspace identification methods; and part (d) further gives the uniqueness of the trajectory by fixing a sufficiently long past trajectory. Part (d) also allows the formulation to be applied in simulation and predictive control.

\subsection{Deterministic Data-Driven Simulation and Control}

In the noise-free case, the system can be simulated solely based on a known trajectory by applying Theorem \ref{thm:1}(d) \cite{Markovsky_2005b}. Define
\begin{equation}
    U_p = \begin{bmatrix}
        u_0^d&u_1^d&\cdots&u_{M-1}^d\\
        \vdots&\vdots&\ddots&\vdots\\
        u_{L_0-1}^d&u_{L_0}^d&\cdots&u_{M+L_0-2}^d\\
        \end{bmatrix}\in\mathbb{R}^{L_0 n_u\times M},
    \label{eqn:Up}
\end{equation}
\begin{equation}
    U_f = \begin{bmatrix}
        u_{L_0}^d&u_{L_0+1}^d&\cdots&u_{M+L_0-1}^d\\
        \vdots&\vdots&\ddots&\vdots\\
        u_{L-1}^d&u_L^d&\cdots&u_{N-1}^d\\
        \end{bmatrix}\in\mathbb{R}^{L' n_u\times M},
    \label{eqn:Uf}
\end{equation}
\begin{equation}
    \mathbf{u}_{\text{ini}} = \text{col}\left(u_{-L_0},\cdots,u_{-1}\right)\in\mathbb{R}^{L_0 n_u},
\end{equation}
\begin{equation}
    \mathbf{u} = \text{col}\left(u_{0},\cdots,u_{L'-1}\right)\in\mathbb{R}^{L' n_u},
\end{equation}
and similarly for $Y_p$, $Y_f$, $\mathbf{y}_{\text{ini}}$, and $\mathbf{y}$. Then we interpret (\ref{eqn:fund}) as an implicit model of the system trajectory parametrized by $g$, namely
\begin{subequations}
  \begin{empheq}[left=\empheqlbrace]{align}
    \mathbf{u}=U_f g, \label{eqn:uu}\\
    \mathbf{y}=Y_f g, \label{eqn:yy}
    \end{empheq}
    \label{eqn:uy}%
\end{subequations}
subject to the initial condition requirement
\begin{subequations}
  \begin{empheq}[left=\empheqlbrace]{align}
    \mathbf{u}_{\text{ini}} = U_p g, \label{eqn:up}\\
    \mathbf{y}_{\text{ini}} = Y_p g. \label{eqn:yp}
    \end{empheq}
    \label{eqn:uyp}%
\end{subequations}
Thus, the system can be simulated by means of a two-step approach with $g$ as the intermediate parameter as shown in Algorithm~\ref{al:1}. The system identification process is performed online for a particular input by estimating $g$. This algorithm effectively gives an implicit model of the system in the form of
\begin{equation}
    \mathbf{y}=f(\mathbf{u};\mathbf{u}_{\text{ini}},\mathbf{y}_{\text{ini}}, U_p,U_f,Y_p,Y_f),
    \label{eqn:ddmodel}
\end{equation}
where $U_p,U_f,Y_p,Y_f$ are offline data that describe the behaviors of the system, and $\mathbf{u}_{\text{ini}}, \mathbf{y}_{\text{ini}}$ are online data that estimate the initial condition.

\begin{algorithm}[htb]
	\caption{Noise-free data-driven simulation \cite{Markovsky_2005b}}
	\begin{algorithmic}[1]
	\State \textbf{Given: }$U_p,U_f,Y_p,Y_f$.
	\State \textbf{Input: }$\mathbf{u}_{\text{ini}},\mathbf{y}_{\text{ini}},\mathbf{u}$.
	\State Solve the linear system
	\begin{equation}
	    \text{col}\left(\mathbf{u}_{\text{ini}},\mathbf{y}_{\text{ini}},\mathbf{u}\right)=\text{col}\left(U_p,Y_p,U_f\right)g
	    \label{eqn:con}
	\end{equation}
	for $g$.
	\State \textbf{Output: }$\mathbf{y}=Y_f g$.
	\end{algorithmic}
	\label{al:1}
\end{algorithm}

Multiple control algorithms have been developed based on this model structure \cite{DePersis_2020,Markovsky_2008}. In this work, we focus on the optimal trajectory tracking problem, which optimizes the following control cost over a horizon of length $L'$ at each time instant $t$ \cite{camacho2016model}:
\begin{equation}
    J_\text{ctr}(\mathbf{u},\mathbf{y})=\sum_{k=0}^{L'-1}\left(\norm{y_k-r_{t+k}}_Q^2+\norm{u_k}_R^2\right),
\end{equation}
where $\mathbf{r}$ is the reference trajectory, and $Q$ and $R$ are the output and the input cost matrices respectively \cite{camacho2016model}. At each time instant, the first entry in the newly optimized input trajectory is applied to the system in a receding horizon fashion.

Algorithm~\ref{al:1} can be applied as the predictor in place of the model-based predictor in conventional MPC algorithms. This leads to the following optimization problem
\begin{equation}
\begin{matrix}
    \underset{\mathbf{u},\mathbf{y},g}{\text{minimize}} &  J_\text{ctr}(\mathbf{u},\mathbf{y})\\
    \text{subject to} & (\ref{eqn:uy}),(\ref{eqn:uyp}),\mathbf{u} \in \pazocal{U}, \mathbf{y} \in \pazocal{Y},
\end{matrix}
\label{eqn:deepc0}
\end{equation}
where $\pazocal{U}$ and $\pazocal{Y}$ are the constraint sets of the inputs and the outputs respectively. Vectors $\mathbf{u}_{\text{ini}}$ and $\mathbf{y}_{\text{ini}}$ are the immediate past input-output measurements online. This method is known as the unregularized DeePC algorithm \cite{Coulson_2019}. % By applying this data-driven model as the predictor in the receding horizon predictive control, we have the unregularized DeePC algorithm \cite{Coulson_2019}, which will be further discussed in Section~\ref{sec:6}.

\section{Maximum Likelihood Data-Driven Model: \Newmodel}
\label{sec:4}

The linear system (\ref{eqn:con}) is highly underdetermined when a large dataset is available. When the data are noise-free, this parameter estimation problem is trivial, where any solution to (\ref{eqn:con}) gives an exact output model of the system, according to Theorem~\ref{thm:1}(d).

However, the problem of finding the model (\ref{eqn:ddmodel}) becomes ill-conditioned when the data are noisy. In this case, Theorem~\ref{thm:1}(c) is no longer satisfied. Instead, $\text{col}(U,Y)$ has full row rank almost surely. If we still follow Algorithm~\ref{al:1}, any output trajectory $\mathbf{y}$ can be obtained by choosing different solutions to (\ref{eqn:con}). In fact, Theorem \ref{thm:1}(d) does not hold exactly for the noisy case, so satisfying condition (\ref{eqn:con}) is not guaranteed to be statistically optimal. An empirical remedy to this problem is to use the Moore-Penrose pseudoinverse solution of $g$, namely
\begin{equation}
    g_{\text{pinv}}=\text{col}\left(U_p,Y_p,U_f\right)^\dagger\text{col}\left(\mathbf{u}_{\text{ini}},\mathbf{y}_{\text{ini}},\mathbf{u}\right),
    \label{eqn:pinv}
\end{equation}
which solves the least-norm problem
\begin{equation}
    \underset{g}{\text{minimize}} \ \norm{g}_2^2 \quad
    \text{subject to} \ (\ref{eqn:con}).
\label{eqn:ln}
\end{equation}
This solution is known as the subspace predictor related to the prediction error method \cite{Huang_2019,Sedghizadeh_2018}. However, this predictor fails to appropriately encode the effects of noise in the data matrices. To the best of our knowledge, there is no existing statistical framework for estimating $g$. In what follows, we will derive a maximum likelihood estimator of $g$. As opposed to existing algorithms, this estimator obtains a statistically optimal data-driven model for systems with noise. Since this model is expressed purely in terms of matrices of signal trajectories, we name this model the \newmodel. For simplicity of exposition, the results in the section are stated for the single-input single-output case, but they seamlessly hold for the multiple-input multiple-output case.

\subsection{Derivation of the Maximum Likelihood Estimator}

In this work, we consider a scenario where the output errors are i.i.d. Gaussian noise for both offline and online data, i.e.,
\begin{equation}
    y_i^d = y_i^{d,0}+w_i^d,\,(w_i^d)_{i=0}^{N-1}\sim \pazocal{N}(0,\sigma^2\mathbb{I}),
\end{equation}
\begin{equation}
    \mathbf{y}_{\text{ini}} = \mathbf{y}_{\text{ini}}^{0}+\mathbf{w}_{p},\,\mathbf{w}_{p}\sim \pazocal{N}(0,\sigma_{p}^2\mathbb{I}).
\end{equation}
\iffalse
\begin{remark}
    The formulation can be extended to other noise models in a straightforward way, including correlated noise, input noise, and alternative noise distributions. For example, when the noise is Laplacian, it would lead to an $l_1$-norm penalization in the estimator similar to the regularizer proposed in \cite{Coulson_2019}.
\end{remark}
\vspace{1em}
\fi
Under this noise model, the equations (\ref{eqn:uu}) and (\ref{eqn:up}) still hold exactly, but the past output equation (\ref{eqn:yp}) includes noise on both sides, which leads to a total least squares problem. In this work, the maximum likelihood interpretation of the total least squares problem is used \cite{Markovsky_2007}.

Define
\begin{equation}
    \hat{\mathbf{y}}=\begin{bmatrix}\epsilon_y\\\mathbf{y}\end{bmatrix}=Yg-\begin{bmatrix}\mathbf{y}_{\text{ini}}\\\mathbf{0}\end{bmatrix},
    \label{eqn:yhat}
\end{equation}
where $\epsilon_y:=Y_pg-\mathbf{y}_{\text{ini}}$ is the residual of the past output relation (\ref{eqn:yp}). Then we want to construct an estimator that maximizes the conditional probability of observing the realization $\hat{\mathbf{y}}$ corresponding to the available data given $g$. Applying vectorization on $Yg$ in (\ref{eqn:yhat}), we have
\begin{equation}
    \hat{\mathbf{y}}=\left(g^\mathsf{T}\otimes \mathbb{I}\right)\text{vec}(Y)-\begin{bmatrix}\mathbf{y}_{\text{ini}}\\\mathbf{0}\end{bmatrix},
\end{equation}
where we make use of the property of the Kronecker product
\begin{equation}
    \text{vec}(ABC) = (C^\mathsf{T}\otimes A)\text{vec}(B).
\end{equation}
Denote the noise-free version of $Y_p$ and $Y_f$ by $Y_p^0$ and $Y_f^0$ respectively. Then for a given $g$, we have
\begin{equation}
\begin{aligned}
    \mathbb{E}(\hat{\mathbf{y}}|g)&=\mathbb{E}(Y)g-\begin{bmatrix}\mathbb{E}(\mathbf{y}_{\text{ini}})\\\mathbf{0}\end{bmatrix}=\begin{bmatrix}Y_p^0g-\mathbf{y}_{\text{ini}}^0\\Y_f^0g\end{bmatrix}=\begin{bmatrix}\mathbf{0}\\Y_f^0g\end{bmatrix},\\
    \text{cov}(\hat{\mathbf{y}}|g)&=\left(g^\mathsf{T}\otimes \mathbb{I}\right)\Sigma_{yd}\left(g\otimes \mathbb{I}\right)+\begin{bmatrix}\sigma_p^2\mathbb{I}&\mathbf{0}\\\mathbf{0}&\mathbf{0}\end{bmatrix}=: \Sigma_y,
    \label{eqn:stats}
\end{aligned}
\end{equation}
where $\Sigma_{yd}=\text{cov}\left(\text{vec}(Y)\right)$. According to the noise model of $\left(y_i^d\right)$ and accounting for the Hankel structure of $Y$, we have
\begin{equation}
    \left(\Sigma_{yd}\right)_{i,j}=\begin{cases}
    \sigma^2,&\left(\text{vec}(Y)\right)_i=\left(\text{vec}(Y)\right)_j\\
    0,&\text{otherwise}
    \end{cases}.
    \label{eqn:ry}
\end{equation}
Then, $\Sigma_y$ can be calculated as
\begin{equation}
    \left(\Sigma_y\right)_{i,j}=\sigma^2\sum_{k=1}^{M-|i-j|}g_k g_{k+|i-j|}+\begin{cases}\sigma_p^2,&i=j\leq L_0\\0,&\text{otherwise}\end{cases}.
    \label{eqn:py}
\end{equation}
where $g_k$ denotes the $k$-th entry of $g$. The derivation is given in Appendix~\ref{sec:app1}. Thus, due to the linearity of the normal distribution, we have
\begin{equation}
\hat{\mathbf{y}}|g\sim \pazocal{N}\left(\begin{bmatrix}\mathbf{0}\\Y_f^0g\end{bmatrix},\Sigma_y\right),
\label{eqn:condprob}
\end{equation}
which has the probability density
\begin{equation}
\begin{split}
    p(\hat{\mathbf{y}}|g)=&(2\pi)^{-\frac{L}{2}}\det{(\Sigma_y)}^{-\frac{1}{2}}\\
    &\exp{\left(-\frac{1}{2}\begin{bmatrix}Y_pg-\mathbf{y}_{\text{ini}}\\Y_fg-Y_f^0g\end{bmatrix}^\mathsf{T}\Sigma_y^{-1}\begin{bmatrix}Y_pg-\mathbf{y}_{\text{ini}}\\Y_fg-Y_f^0g\end{bmatrix}\right)}.
\end{split}
    \label{eqn:prob1}
\end{equation}
Note that here the true output data matrix $Y_f^0$ is also unknown, and can be estimated with the maximum likelihood approach.
In this way, we are ready to derive the \newmodel\ by solving the following optimization problem.
\begin{equation}
    \underset{g\in\pazocal{G},Y_f^0}{\text{minimize}}\ -\log p(\hat{\mathbf{y}}|g,Y_f^0),
    \label{eqn:opt1}
\end{equation}
where $\pazocal{G} = \left\{g\in \mathbb{R}^M\left|\,\text{col}\left(U_p,U_f\right)g=\text{col}\left(\mathbf{u}_{\text{ini}},\mathbf{u}\right)\right.\right\}$ is the parameter space defined by the known noise-free input trajectory.

Substituting (\ref{eqn:prob1}) into (\ref{eqn:opt1}), we have the equivalent optimization problem,
\begin{equation}
\begin{split}
    \underset{g\in\pazocal{G},Y_f^0}{\text{minimize}}\ &\text{logdet}(\Sigma_y(g))\\[-1em]
    &+\begin{bmatrix}Y_pg-\mathbf{y}_{\text{ini}}\\Y_fg-Y_f^0g\end{bmatrix}^\mathsf{T}\Sigma_y^{-1}(g)\begin{bmatrix}Y_pg-\mathbf{y}_{\text{ini}}\\Y_fg-Y_f^0g\end{bmatrix}.
\end{split}
\label{eqn:opt2}
\end{equation}
It is easy to see that the optimal value of $Y_f^0$ is $Y_f$ regardless of the choice of $g$. So (\ref{eqn:opt2}) is equivalent to
\begin{equation}
    \underset{g\in\pazocal{G}}{\text{minimize}}\ \text{logdet}(\Sigma_y(g))+\begin{bmatrix}Y_pg-\mathbf{y}_{\text{ini}}\\\mathbf{0}\end{bmatrix}^\mathsf{T}\Sigma_y^{-1}(g)\begin{bmatrix}Y_pg-\mathbf{y}_{\text{ini}}\\\mathbf{0}\end{bmatrix}.
\label{eqn:opt0}
\end{equation}
In this objective function, the first term indicates how accurate the output estimates are. The second term represents how much the estimate deviates from the past output observations.

\subsection{Iterative Computation of the Estimator}

To find a computationally efficient algorithm to solve (\ref{eqn:opt0}), we relax the problem and solve it with sequential quadratic programming (SQP) \cite{boggs1995sequential}. First, the covariance matrix $\Sigma_y$ is approximated with its diagonal part, denoted by $\bar{\Sigma}_y$, i.e.,
\begin{equation}
    \left(\bar{\Sigma}_y\right)_{i,j} = \begin{cases}
    \left(\Sigma_y\right)_{i,j},&i=j\\
    0,&i\neq j
    \end{cases}.
\end{equation}

\begin{remark}
    This approximation holds exactly when the data matrices are constructed by truncating $(u_i^d, y_i^d)_{i=0}^{N-1}$ into sections of length $L$ with no overlap, or using multiple independent trajectories of length $L$, instead of forming Hankel structures. This construction is known as the Page matrix \cite{damen1982approximate} and it was shown in \cite{vanWaarde_2020} that similar results to Theorem~\ref{thm:1} still hold for Page matrices. The Hankel construction is able to use the data more efficiently, but leads to complex noise correlation, which is reflected in the non-diagonal structure of $\Sigma_y$. The comparison between the Hankel construction and the Page construction is, however, beyond the scope of this paper. See \cite{IM-FD,iannelli2021design} for more on this topic.
\end{remark}

\begin{remark}
    This approximation gives an upper bound on the log-det terms. According to Hadamard's inequality, since $\Sigma_y\in \mathbb{S}_{++}^{L}$, we have $\mathrm{logdet}(\Sigma_y(g))\leq\mathrm{logdet}(\bar{\Sigma}_y(g))$.
\end{remark}

\vspace{1em}
In this way, problem (\ref{eqn:opt0}) is approximated as
\begin{equation}
\begin{split}
    \underset{g\in\pazocal{G}}{\text{minimize}}\ L'\log&\left(\norm{g}_2^2\right)+L_0\log\left(\sigma^2\norm{g}_2^2+\sigma_p^2\right)\\
    &\quad\quad\ +\dfrac{1}{\sigma^2\norm{g}_2^2+\sigma_p^2}\norm{Y_pg-\mathbf{y}_{\text{ini}}}_2^2.
\end{split}
\label{eqn:optapp}
\end{equation}
This problem can be readily solved by SQP. For each iteration, the following quadratic programming problem is solved.
\begin{equation}
\begin{matrix}
    \quad\quad g^{k+1}=&
    \text{arg}\underset{g}{\text{min}} &    \lambda(g^k)\norm{g}_2^2+\norm{Y_pg-\mathbf{y}_{\text{ini}}}_2^2 \\
    &\text{subject to} & \begin{bmatrix}
    U_p\\U_f
    \end{bmatrix}g=\begin{bmatrix}
    \mathbf{u}_{\text{ini}}\\\mathbf{u}
    \end{bmatrix},
\end{matrix}
\label{eqn:optsqp}
\end{equation}
where $\lambda(g^k) = L'\sigma_p^2\big{/}\norm{g^k}_2^2+L \sigma^2$.
The objective function in (\ref{eqn:optapp}) is approximated by a quadratic function around $g^k$, making use of the local expansion $\log x\approx \log x_0 + \frac{1}{x_0}(x-x_0)$. The optimality conditions of (\ref{eqn:optsqp}) are: 
\begin{equation}
    \begin{bmatrix}
    F(g^k)&U^\mathsf{T}\\
    U&\mathbf{0}
    \end{bmatrix}\begin{bmatrix}g^{k+1}\\\nu^{k+1}\end{bmatrix}=
    \begin{bmatrix}Y_p^\mathsf{T}\mathbf{y}_{\text{ini}}\\\tilde{\mathbf{u}}\end{bmatrix},
\end{equation}
where $\tilde{\mathbf{u}}=\text{col}(\mathbf{u}_{\text{ini}},\mathbf{u})$, $F(g^k)=\lambda(g^k)\mathbb{I}+Y_p^\mathsf{T}Y_p$, and $\nu_{k+1}\in\mathbb{R}^L$ is the Lagrange multiplier. The closed-form solution is thus given by
\begin{equation}
    \begin{aligned}
    g^{k+1} &= \left(F^{-1}-F^{-1}U^\mathsf{T}(U F^{-1}U^\mathsf{T})^{-1}UF^{-1}\right)Y_p^\mathsf{T}\mathbf{y}_{\text{ini}}\\&\quad\quad\quad\quad\quad\quad\quad\quad\quad\quad+F^{-1}U^\mathsf{T}(U F^{-1}U^\mathsf{T})^{-1}\tilde{\mathbf{u}}\\
    &=:\pazocal{P}(g^k)\mathbf{y}_{\text{ini}}+\pazocal{Q}(g^k)\tilde{\mathbf{u}}.
    \end{aligned}
    \label{eqn:clsol}
\end{equation}
This algorithm converges to a local minimum of problem (\ref{eqn:optapp}). %In our simulation, this iterative algorithm usually converges in two to four steps.

\begin{remark}
    Following similar derivations, this algorithm can be extended to the case where i.i.d. Gaussian input errors also exist in offline and online data,, which leads to an additional input regularization term $\norm{Ug-\tilde{\mathbf{u}}}_2^2$ in the iterative algorithm.
    \iffalse
    In this case, the iterative algorithm solves optimization problems in the following form:
    \begin{equation}
    \resizebox{0.49\textwidth}{!}{$
    g^{k+1}=
    \text{arg}\underset{g}{\text{min}}\  \norm{g}_2^2+\lambda_1(g^k)\norm{Y_pg-\mathbf{y}_{\text{ini}}}_2^2+\lambda_2(g^k)\norm{Ug-\tilde{\mathbf{u}}}_2^2$}
    \end{equation}
    \fi
\end{remark}
\subsection{Maximum Likelihood Data-Driven Simulation}

Based on the derived maximum likelihood estimator of $g$, the step of solving the linear system (\ref{eqn:con}) in Algorithm~\ref{al:1} can be replaced by solving the SQP problem (\ref{eqn:optsqp}). For simulation, the SQP problem can be initialized at the pseudoinverse solution $g_\text{pinv}$ (\ref{eqn:pinv}). This leads to the following algorithm for maximum likelihood data-driven simulation.

\begin{algorithm}[htb]
	\caption{Maximum likelihood data-driven simulation: the \newmodel}
	\begin{algorithmic}[1]
	\State \textbf{Given: }$U_p,U_f,Y_p,Y_f,\sigma,\sigma_p,\epsilon$.
	\State \textbf{Input: }$\mathbf{u}_{\text{ini}},\mathbf{y}_{\text{ini}},\mathbf{u}$.
	\State $k \leftarrow 0$, $g^0 \leftarrow g_{\text{pinv}}$ from (\ref{eqn:pinv})
      \Repeat
        \State Calculate $g^{k+1}$ with (\ref{eqn:clsol}).
        \State $k \leftarrow k+1$
      \Until{\norm{g^k-g^{k-1}}<\epsilon\norm{g^{k-1}}}
	\State \textbf{Output: }$g_{\text{\newacro}}=g^k$, $\mathbf{y}=Y_f g^k$.
	\end{algorithmic}
	\label{al:2}
\end{algorithm}

This algorithm gives the \newmodel\ in the form of (\ref{eqn:ddmodel}). The approximate maximum likelihood estimator (\ref{eqn:optsqp}) has the same $\norm{g}_2^2$-penalization term as the least-norm problem (\ref{eqn:ln}). However, the estimate $g_{\text{\newacro}}$ does not lie in the solution space of the underdetermined system (\ref{eqn:con}). The total least squares structure in (\ref{eqn:yp}) leads to the penalization term $\norm{Y_pg-\mathbf{y}_{\text{ini}}}_2^2$ in place of the hard constraint in (\ref{eqn:con}).

\section{Practical Aspects and Analysis of the \Newmodel}
\label{sec:4.5}

In this section, we first discuss practical scenarios where the dimension of the signal matrix is very large and the noise level information required in formulating the \newmodel\ is unknown. Then, the effectiveness of the proposed maximum likelihood framework is analyzed by numerical comparison and covariance analysis.

\subsection{Preconditioning of Data Matrices}
\label{sec:precon}
In data-driven applications, it is usually assumed that abundant data are available, i.e., $N\gg L$. Under this scenario, the dimension of the parameter vector $g\in \mathbb{R}^M$, which needs to be optimized online, would be much larger than the length of the predicted output trajectory. This leads to high online computational complexity even to estimate a very short trajectory. On the other hand, only $2L$ independent basis vectors are needed to describe all the possible input-output trajectories of length $L$. It is possible to precondition the data matrices such that only $2L$ basis trajectories are used. 

To do this, we propose the following strategy based on the SVD to compress the data such that the dimension of the parameter vector $g$ is $2L$ regardless of the raw data length. Let
$
    \text{col}(U,Y)=WSV^\mathsf{T}\in\mathbb{R}^{2L\times M}
$
be the SVD of the data matrix. Define the compressed data matrices $\tilde{U}_p,\tilde{Y}_p\in\mathbb{R}^{L_0\times 2L}$ and $\tilde{U}_f,\tilde{Y}_f\in\mathbb{R}^{L'\times 2L}$ such that
\begin{equation}
    \text{col}\left(\tilde{U}_p,\tilde{U}_f,\tilde{Y}_p,\tilde{Y}_f\right)=WS_{2L}\in\mathbb{R}^{2L\times 2L},
\end{equation}
where $S_{2L}$ is the first $2L$ columns of $S$.

It is shown in the following proposition that Algorithm~\ref{al:2} with compressed data matrices obtains exactly the same output trajectory $\mathbf{y}$ as with raw data matrices.

\begin{prop}
Let the simulated trajectories with data matrices $(U_p, Y_p, U_f, Y_f)$ and $(\tilde{U}_p, \tilde{Y}_p, \tilde{U}_f, \tilde{Y}_f)$ from Algorithm~\ref{al:2} be $\mathbf{y}$ and $\tilde{\mathbf{y}}$ respectively. Then we have $\tilde{\mathbf{y}}=\mathbf{y}$.
\label{prop:1}
\end{prop}

\begin{proof} Define transformed data matrices $\bar{U}_p$, $\bar{Y}_p$, $\bar{U}_f$, and $\bar{Y}_f$ such that $\text{col}\left(\bar{U}_p, \bar{Y}_p, \bar{U}_f, \bar{Y}_f\right)=WS$.
Then the relations between the data matrices are given by
\begin{equation}
\begin{aligned}
    \text{col}\left(U_p, Y_p, U_f, Y_f\right)&=\text{col}\left(\tilde{U}_p,\tilde{U}_f,\tilde{Y}_p,\tilde{Y}_f\right)V_{2L}^\mathsf{T},\\
    \text{col}\left(U_p, Y_p, U_f, Y_f\right)&=\text{col}\left(\bar{U}_p, \bar{Y}_p, \bar{U}_f, \bar{Y}_f\right)V^\mathsf{T},\\
    \text{col}\left(\bar{U}_p, \bar{Y}_p, \bar{U}_f, \bar{Y}_f\right)&=\left[\text{col}\left(\tilde{U}_p,\tilde{U}_f,\tilde{Y}_p,\tilde{Y}_f\right)\ \mathbf{0}\right].
\end{aligned}
\end{equation}
where $V_{2L}$ denotes the first $2L$ columns of $V$. 

Denote the variables with the compressed data matrices by a tilde, and the variables with the transformed data matrices by a bar. Since $V_{2L}^\mathsf{T}V_{2L}=\mathbb{I}$, we have $g_{\text{pinv}}=V_{2L}\,\tilde{g}_{\text{pinv}}$.
This leads to $\norm{g_{\text{pinv}}}_2^2=\norm{\tilde{g}_{\text{pinv}}}_2^2$, and thus $\lambda(g^0)=\lambda(\tilde{g}^{0})$.

Suppose at the $k$-th iteration, $\lambda(g^k)=\lambda(\tilde{g}^{k})$. Due to the orthogonality of $V$ and the sparsity structure of $\bar{U}$ and $\bar{Y}_p$, we have $g^{k+1} = V\bar{g}^{k+1},\,\bar{g}^{k+1}=\text{col}\left(\tilde{g}^{k+1},\mathbf{0}\right)$.
This leads to
\begin{equation}
    g^{k+1} = V_{2L}\,\tilde{g}^{k+1},\,\norm{g^{k+1}}_2^2=\norm{\tilde{g}^{k+1}}_2^2.
    \label{eqn:proofnorm}
\end{equation}
Thus for all $k$, we have $g^{k} = V_{2L}\,\tilde{g}^{k}$. Therefore, the simulated trajectory satisfies $\mathbf{y}=Y_f g^k=\tilde{Y}_f \tilde{g}^{k}=\tilde{\mathbf{y}}$.
\end{proof}

\begin{remark}
    It can be seen from the proof that $\bar{\Sigma}_y(g)=\bar{\Sigma}_y(\tilde{g})$. So with compressed data matrices, the output trajectory estimate has the same covariance as the raw data matrices when Page matrices are used, and the same diagonal components when Hankel matrices are used.
\end{remark}

%The \newmodel\ (Algorithm~\ref{al:2}), together with the above data compression scheme, can be applied to various problems in system identification and control. In the following two sections, two approaches are discussed to apply this model to problems in impulse response estimation and receding horizon predictive control.

\subsection{Data-driven Noise Level Estimation}
\label{sec:sigest}

When the noise level $\sigma^2$ in the output signal matrix $Y$ is unknown, it can be directly estimated from the singular values of a projected signal matrix. In detail, let $\Pi_U^\perp=\mathbb{I}-U^\top(UU^\top)^{-1}U$ span the null space of the noise-free matrix $U$, and $Y^0$ be the noise-free version of $Y$. Then according to the persistency of excitation requirement and Theorem~\ref{thm:1}(c), $\text{rank}(Y^0 \Pi_U^\perp)=n_x$, and thus $Y\Pi_U^\perp$ is a perturbed rank-$n_x$ matrix. We apply the robust noise level estimator for perturbed low-rank matrices presented in Section~III-E of \cite{Gavish_2014}
\begin{equation}
    \hat{\sigma}^2=\frac{s_\text{med}^2}{M\mu(L/M)},
    \label{eqn:sigest}
\end{equation}
where $s_\text{med}$ is the median of the singular values of $Y\Pi_U^\perp$ and $\mu(\beta)$ is the median of the Marchenko-Pastur distribution with aspect ratio $\beta$.
\iffalse
solved by the equation
\begin{equation}
    \int_{(1-\sqrt{\beta})^2}^{\mu(\beta)}\dfrac{\sqrt{\left((1+\sqrt{\beta})^2-t\right)\left(t-(1-\sqrt{\beta})^2\right)}}{2\pi t}\mathrm{d}t=\dfrac{1}{2}.
\end{equation}
\fi
This estimator compares the perturbed singular values with the ideal distribution of the noise singular values to estimate $\sigma^2$. It is applicable when $s_\text{med}$ comes purely from noise, i.e., $L>2n_x$.

The noise level of online data $\sigma_p^2$ can be set to zero when initial conditions are known exactly, or to $\sigma^2$ when the same sensor is used for offline and online measurements. Otherwise, online measurements can be taken beforehand and $\sigma_p^2$ can be estimated similarly to $\sigma^2$.

\subsection{Comparison of Data-Driven Predictors}
\label{sec:comp}
The performance of Algorithm~\ref{al:2} is analyzed numerically by comparing the accuracy of the simulated output $\mathbf{y}$ measured by fitting metric
\begin{equation}
    W=100\cdot \left(1-\left[\frac{\sum_{i=1}^{L'}(y_i-\hat{y}_i)^2}{\sum_{i=1}^{L'}(y_i-\bar{y})^2}\right]^{1/2}\right),
    \label{eq:W}
\end{equation}
where $y_i$ are the true outputs, $\hat{y}_i$ are the estimated outputs, and $\bar{y}$ is the mean of the true outputs. We compare 1) \textit{pinv}: the least-norm solution (\ref{eqn:pinv}), 2) \textit{exact}: the SQP solution of problem (\ref{eqn:opt0}) initialized at $g_{\text{pinv}}$, 3) \textit{SMM-1}: the solution after one iteration of Algorithm~\ref{al:2}, and 4) \textit{SMM}: Algorithm~\ref{al:2}.

Consider random single-input single-output systems with state dimensions between 2 and 10 (generated by \textsc{Matlab} function \texttt{drss}). The following parameters are used: $L_0=n_x$, $L'=10$. Inputs for the identification data $(u_i^d)_{i=0}^{N-1}$ and simulation conditions $\mathbf{u}_{\text{ini}}$, $\mathbf{y}_{\text{ini}}$, $\mathbf{u}$ are all unit i.i.d. Gaussian. For each analysis, 100 Monte Carlo simulations are conducted.

The simulation accuracy of different MLE algorithms are plotted in Figure~\ref{fig:0.1}(a) for different data sizes $N$. It can be seen that for small data sizes, the \textit{exact} estimate obtains very similar performance to the \newacro\ estimates. This indicates that the approximate solution obtains a close match to the original MLE problem. For larger data sizes, due to the increasing dimension of $g$, the performance of \textit{exact}, where the data compression scheme does not apply, becomes worse. On the other hand, Algorithm~\ref{al:2} converges very quickly as the one-iteration solution \textit{SMM-1} obtains almost identical performance to the converged solution \textit{SMM} at all data sizes.

\begin{figure}[htb]
\centering
\begin{tabular}{ c @{\hspace{5pt}} c }
\includegraphics[width=1.65in]{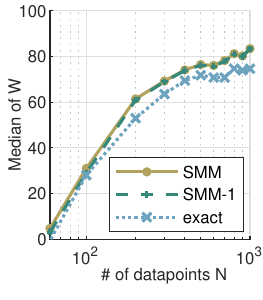} &
\includegraphics[width=1.65in]{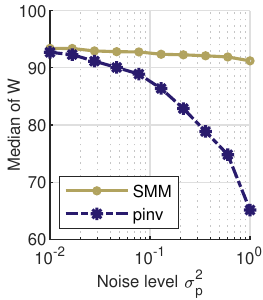} \\[-0.5em]
\footnotesize (a) $\sigma^2=\sigma_p^2=1$&
\footnotesize (b) $N=100$, $\sigma^2=0.01$
\end{tabular}
\caption{Comparison of simulation accuracy with different data-driven predictors.}
\label{fig:0.1}
\end{figure}

The \textit{SMM} estimate is compared against \textit{pinv} in Figure~\ref{fig:0.1}(b) for different online noise levels $\sigma_p^2$. It is showcased that \textit{SMM} is more accurate than \textit{pinv} due to the inclusion of the correct noise model. In particular, this performance improvement is more significant when $\sigma_p^2$ is large. In order to assess the general validity of the results shown in Figure~\ref{fig:0.1}(b), it is demonstrated theoretically in the following proposition that the signal matrix model obtains a smaller covariance than the least-norm solution when noise is present only in $\mathbf{y}_\text{ini}$.
\begin{prop}
Let $g_{\text{pinv}}$ and $g_{\text{\newacro}}$ be the estimates from the least-norm solution (\ref{eqn:pinv}) and Algorithm~\ref{al:2} respectively. When $\sigma^2=0$, we have $\text{tr}(\text{cov}(g_{\text{\newacro}}))<\text{tr}(\text{cov}(g_{\text{pinv}}))$.
\label{prop:2}
\end{prop}
\begin{proof}
    See Appendix~\ref{sec:app2}.
\end{proof}

\section{Impulse Response Estimation With the \Newmodel}
\label{sec:5}

We propose here a system identification method that identifies an FIR model of the system by \newmodel\ simulation. Numerical tests show that model fitting is improved compared to the conventional least-squares estimate, when the truncation error is large or the input history is unknown.

\subsection{Impulse Response Estimation}

Consider the problem of estimating the impulse response model $(h_i)_{i=0}^\infty$ of a system from data. The output $y_t$ is given by
\begin{equation}
y_t = \sum^\infty_{i=0} h_i u_{t-i}.
\end{equation}
The conventional approach is to formulate a linear regression to estimate a finite truncation of the impulse response
\begin{equation}
    \underbrace{\begin{bmatrix}y^d_0\\y^d_1\\\vdots\\y^d_{N-1}\end{bmatrix}}_{y_N}=\underbrace{\begin{bmatrix}u^d_0&u^d_{-1}&\cdots&u^d_{1-n}\\u^d_1&u^d_0&\cdots&u^d_{2-n}\\\vdots&\vdots&\ddots&\vdots\\u^d_{N-1}&u^d_{N-2}&\cdots&u^d_{N-n}\end{bmatrix}}_{\Phi_N}\underbrace{\begin{bmatrix}h_0\\h_1\\\vdots\\h_{n-1}\end{bmatrix}}_{h},
\end{equation}
where $n$ is the length of the impulse response to be estimated. The regression problem can then be solved by least squares with the closed-form solution
\begin{equation}
    \hat{h}_\text{LS} = \left(\Phi_N^\mathsf{T}\Phi_N\right)^{-1}\Phi_N^\mathsf{T}y_N.
    \label{eqn:LS}
\end{equation}
There are two main assumptions underlying this formulation: 1) the truncation error of the finite impulse response is negligible, i.e., $h_i\approx 0$ for all $i\geq n$; and 2) additional input measurements $(u_i^d)_{i=1-n}^{-1}$ are available. With these two assumptions, the least-squares solution is known to be the best unbiased estimator with i.i.d. Gaussian output noise \cite{LjungBook2}.

However, these assumptions may not be satisfied in practice. When the internal dynamics matrix $A$ has a large condition number, a very long impulse response sequence is needed to remove the truncation error even for a low-order system. In this case, the least-squares algorithm may become impractical due to limited data length and/or computation power. If the truncation error is not negligible, the estimator is not correct, i.e., in the noise-free case, the estimate does not coincide with the true system. When the input history is unknown, the first $(n-1)$ input measurements have to be used solely for initial condition estimation, in which case the data efficiency is substantially affected when a large $n$ is needed.% Similar assumptions are required for non-parametric methods in the frequency domain as well, where periodic input signals are assumed and a sufficiently long dataset is needed to avoid aliasing.

In this work, we propose using the \newmodel\ to estimate the impulse response by finding the length-$n$ response to a pulse input from zero initial conditions, i.e.,
\begin{equation}
    \mathbf{u}_{\text{ini}}=\mathbf{0},\,\mathbf{y}_{\text{ini}}=\mathbf{0},\,\mathbf{u}=\text{col}(1,\mathbf{0}),\,L'=n.
    \label{eqn:ddimp}
\end{equation}
Since the initial condition is known exactly, we have $\sigma_p=0$. Then the output trajectory $\mathbf{y}$ is an estimate of the impulse response $h$ of length $n$ of the system \cite{Markovsky_2005b}. This approach requires neither of the assumptions for the least-squares method. Instead of requiring a length-$(n-1)$ input history sequence, this approach only uses the first $L_0$ entries of the data to estimate the initial condition. In fact, the estimator is correct and unbiased for an arbitrary length $n$ and unknown input history as shown in Theorem~\ref{thm:1}(d), as long as the persistency of excitation condition is satisfied.

\subsection{Numerical Results}

In this subsection, the proposed algorithm is tested against the least-squares estimate by applying it to numerical examples.
\iffalse
The stable spline (SS) kernel structure 
\begin{equation}
    \left(\Sigma_k\right)_{i,j} = \beta\left(\frac{\alpha^{\max (i,j)+i+j}}{2}-\frac{\alpha^{3\max (i,j)}}{6}\right),\,\eta=\text{col}(\alpha,\beta)
\end{equation}
is used \cite{Pillonetto_2010}. 
We compare the following four algorithms: 1) \textit{LS}: least-squares estimate (\ref{eqn:LS}), 2) \textit{LS-SS}: kernel-based/least-squares estimate (\ref{eqn:lsker}) with the SS kernel, 3) \textit{\newacro}: \newmodel\ estimate (Algorithm~\ref{al:2} with (\ref{eqn:ddimp})), and 4) \textit{\newacro-SS}: kernel-based/\newmodel\ estimate (Algorithm~\ref{al:3}) with the SS kernel.
\fi
We compare the proposed \newmodel\ estimate \textit{\newacro} (Algorithm~\ref{al:2} with (\ref{eqn:ddimp})) with the least-squares estimate \textit{LS} (\ref{eqn:LS}). The parameters used in the simulation are 
$
    N=50,\,L_0=4,\,n=L'=11,\,\sigma^2=0.01.
$
In \textit{\newacro}, the noise level $\sigma^2$ is estimated using (\ref{eqn:sigest}). The identification data are generated with unit i.i.d. Gaussian input signals. For each case, 1000 Monte Carlo simulations are conducted.

In the first example, we consider the following fourth-order LTI system tested in \cite{Pillonetto_2010}
\begin{equation}
    G_1(z) = \dfrac{0.1159(z^3+0.5z)}{z^4-2.2z^3+2.42z^2-1.87z+0.7225}.
    \label{eqn:sys1}
\end{equation}
This system is relatively slow. The truncation error is significant when $n=11$ is selected. First, the \textit{LS} and \textit{\newacro} algorithms are compared under the noise-free case, and the results are shown in Figure~\ref{fig:1}(a). It can be clearly seen that \textit{LS} is not correct due to the presence of truncation errors, whereas the \newacro\ estimator is correct. When the noise is present, the \textit{LS} and \textit{\newacro} algorithms are compared in Figure~\ref{fig:1}(b). The \textit{\newacro} estimator have smaller variance compared to \textit{LS}.

\begin{figure}[htb]
\centerline{\includegraphics[width=\columnwidth]{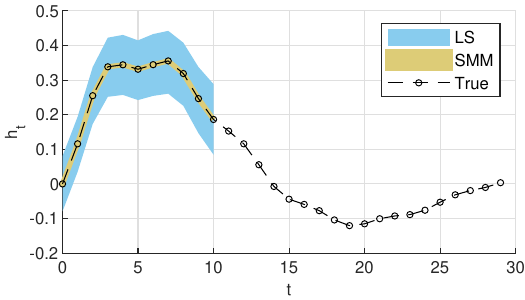}}\vspace{-0.5em}
\centerline{\footnotesize (a) Noise-free case}
\centerline{\includegraphics[width=\columnwidth]{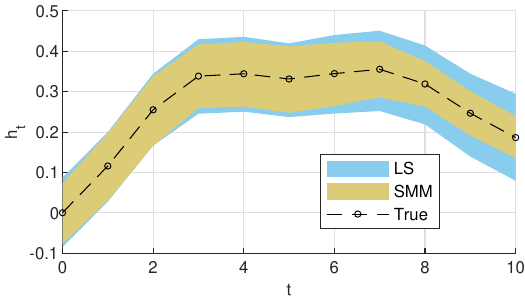}}\vspace{-0.5em}
\centerline{\footnotesize (b) Noisy case with $\sigma^2=0.01$}
\caption{Comparison of impulse response estimation with truncation errors. Colored area: estimates within two standard deviations.}
\label{fig:1}
\end{figure}

In the second example, we focus on the effect of unknown input history by investigating a faster LTI system used in \cite{Pillonetto_2010}
\begin{equation}
    G_2(z) = \dfrac{0.9183z}{z^2+0.24z+0.36}.
\end{equation}
In this case, the truncation error is already negligible at $n=11$, but we assume the input history is unknown. The results of the estimation are illustrated in Figure~\ref{fig:2}. The result of the \textit{\newacro} algorithm is shown to be more accurate than the \textit{LS} algorithm, especially for the first four coefficients.% Regularizing with the SS kernel prior improves the estimation quality for the tail of the impulse response as shown in Figure~\ref{fig:2}(b).

\begin{figure}[htb]
\centerline{\includegraphics[width=\columnwidth]{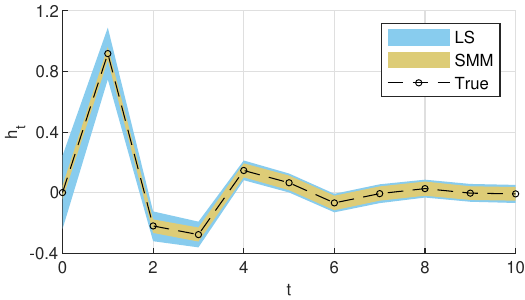}}\vspace{-0.5em}
\iffalse
\centerline{\footnotesize (a) Comparison between \newacro\ and least-squares estimates}\vspace{0.5em}
\centerline{\includegraphics[width=\columnwidth]{4new.eps}}\vspace{-0.5em}
\centerline{\footnotesize (b) Effect of kernel regularization on the \newacro\  method}
\fi
\caption{Comparison of impulse response estimation with unknown input history. Colored area: estimates within two standard deviations.}
\label{fig:2}
\end{figure}

To quantitatively assess the performance of different algorithms, we quantify the model fitting by the metric $W$ (\ref{eq:W}) with impulse response estimates. The boxplots of model fitting for both examples are plotted in Figure~\ref{fig:3}. For comparison, the case with known input history is also plotted for example 2. The \textit{\newacro} algorithm performs better than the \textit{LS} algorithm when the truncation error is large or the input history is unknown. In example 1, the \textit{LS} model fitting is similar for the noisy and noise-free cases, which indicates that the truncation error is the main source of error here. However, when both assumptions of the least squares are satisfied, \textit{LS} performs slightly better than \textit{\newacro}. This is due to the fact that part of the data is used to estimate the initial condition in Algorithm~\ref{al:2}, whereas it is known for the \textit{LS} algorithm.

\begin{figure}[htb]
\centering
  \begin{tabular}{ c @{\hspace{5pt}} c }
    \includegraphics[width=1.65in]{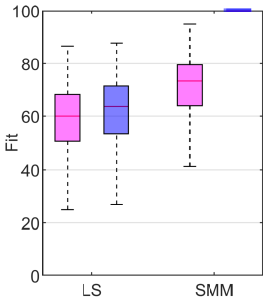} &
    \includegraphics[width=1.65in]{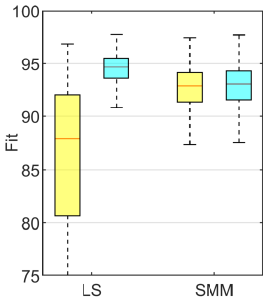} \\
    \footnotesize (a) Example 1&
    \footnotesize (b) Example 2
  \end{tabular}
\caption{Box plots of model fitting for both examples with 1000 simulations. In (a), magenta: noisy data, blue: noise-free data. In (b), yellow: unknown input history, cyan: known input history.}
\label{fig:3}
\end{figure}

\section{Data-Driven Predictive Control With the Signal Matrix Model}
\label{sec:6}

In this section, the \newmodel\ is used as the predictor in receding horizon predictive control. As discussed in Section~\ref{sec:4}, the predictor in the unregularized DeePC problem (\ref{eqn:deepc0}) becomes ill-conditioned when noise is present. In the following subsections, we will present two existing methods to remedy this problem, and compare the control performance with the optimal predictor proposed in this paper.

\subsection{Pseudoinverse and Regularized Algorithms}

There are mainly two types of existing algorithms to extend (\ref{eqn:deepc0}) to the noisy case: the data-driven subspace predictive control and the regularized DeePC algorithm.

The subspace predictive control approach \cite{Sedghizadeh_2018} uses the pseudoinverse solution of $g$ in the predictor instead of the underdetermined linear equality constraints as follows
\begin{equation}
\begin{aligned}
    \underset{\mathbf{u},\mathbf{y}}{\text{minimize}} & \quad\quad\quad\quad\quad\quad\quad J_\text{ctr}(\mathbf{u},\mathbf{y})\\
    \text{subject to} &\quad \mathbf{y}=Y_f\, g_\text{pinv}(\mathbf{u;\mathbf{u}_{\text{ini}},\mathbf{y}_{\text{ini}}}),\mathbf{u} \in \pazocal{U}, \mathbf{y} \in \pazocal{Y},
\end{aligned}
\label{eqn:subpc}
\end{equation}
where $g_\text{pinv}(\cdot)$ is defined in (\ref{eqn:pinv}). Multiple applications have been studied with similar algorithms (e.g., \cite{Hallouzi_2008,Kadali_2003}). However, as discussed in Section~\ref{sec:4}, $g_\text{pinv}(\cdot)$ is not guaranteed to be an effective choice of $g$ for systems with noise.

The regularized DeePC algorithm \cite{Coulson_2019} adds additional regularization terms in the objective in order to regularize both the norm of $g$ and the slack variable needed to satisfy (\ref{eqn:yp}):
\begin{equation}
\begin{aligned}
    \underset{\mathbf{u},\mathbf{y},g,\hat{\mathbf{y}}_{\text{ini}}}{\text{minimize}} &  \quad  J_\text{ctr}(\mathbf{u},\mathbf{y})+\lambda_g\norm{g}_p^p+\lambda_y\norm{\hat{\mathbf{y}}_{\text{ini}}-\mathbf{y}_{\text{ini}}}_p^p \\
    \text{subject to} &\quad\text{col}\left(
    U_p,Y_p,U_f,Y_f
    \right)g=\text{col}\left(
    \mathbf{u}_{\text{ini}},\hat{\mathbf{y}}_{\text{ini}},\mathbf{u},\mathbf{y}
    \right),\\
    &\quad\mathbf{u} \in \pazocal{U}, \mathbf{y} \in \pazocal{Y},
\end{aligned}
\label{eqn:deepc}
\end{equation}
where $p$ is usually selected as 1 or 2. This algorithm can be interpreted as an MPC algorithm acting on the implicit parametric model structure (\ref{eqn:uy}) and (\ref{eqn:uyp}), where the objective is a trade-off between the control performance objective $J_{\text{ctl}}$ and the parameter estimation objective
\begin{equation}
    J_{\text{id,DeePC}}:=\lambda\norm{g}_p^p+\norm{\hat{\mathbf{y}}_{\text{ini}}-\mathbf{y}_{\text{ini}}}_p^p,\,\lambda = \lambda_g/\lambda_y.
    \label{eqn:iddeepc}
\end{equation}
The set of underdetermined model parameters $(g,\hat{\mathbf{y}}_{\text{ini}})$ are then estimated adaptively in the MPC algorithm. The estimated trajectory in this algorithm is not associated with a fixed input-output mapping, but is biased towards those that predict better control performance. This algorithm is also shown to be effective in multiple applications (e.g., \cite{Coulson_2019_reg,Huang_2019}). However, %there are two main problems associated with it. First, the model parameters are optimized in an optimistic manner in that as regularization terms, the estimator is biased towards those that predict better control performance. However, the system behaviors should be independent of the control task. Second, 
tuning of the regularization parameters is a very hard problem. To the best of our knowledge, there is no practical approach proposed to tune $\lambda_g$ and $\lambda_y$ beforehand, and unfortunately the control performance is known to be very sensitive to the regularization parameters \cite{Huang_2019}.

\begin{remark}
    The same data compression scheme as discussed in \ref{sec:precon} is applicable to these two algorithms as well.
\end{remark}

\subsection{An Optimal Tuning-Free Approach}

To address the concerns regarding the two existing methods discussed in the previous subsection, we propose a receding horizon predictive control scheme with the \newmodel\ as the predictor. This directly leads to 
\begin{equation}
\begin{aligned}
    \underset{\mathbf{u},\mathbf{y}}{\text{minimize}} & \quad\quad\quad\quad\quad\quad\quad J_\text{ctr}(\mathbf{u},\mathbf{y})\\
    \text{subject to} &\quad \mathbf{y}=Y_f\, g_\text{\newacro}(\mathbf{u;\mathbf{u}_{\text{ini}},\mathbf{y}_{\text{ini}}}),\mathbf{u} \in \pazocal{U}, \mathbf{y} \in \pazocal{Y},
\end{aligned}
\label{eqn:mlepc0}
\end{equation}
where $g_\text{\newacro}(\cdot)$ is obtained by Algorithm~\ref{al:2}. However, unlike the pseudoinverse predictor where $g_{\text{pinv}}(\cdot)$ is linear with respect to $\mathbf{u}$, the maximum likelihood predictor $g_\text{\newacro}(\cdot)$ involves an iterative algorithm which cannot be expressed as an equality constraint explicitly.

To solve this problem, we notice that the $l_2$-norm of $g$ does not change much throughout the receding horizon control, and the algorithm is only iterative with respect to $\norm{g}_2^2$. So in a receding horizon context, it makes sense to warm-start the optimization problem from the $\norm{g}_2^2$-value at the previous time instant. Then, $g_{\text{\newacro}}(\cdot)$ can be closely approximated by the solution of (\ref{eqn:clsol}) after the first iteration, i.e.,
\begin{equation}
    g^t(\mathbf{u;\mathbf{u}_{\text{ini}},\mathbf{y}_{\text{ini}}},g^{t-1})=\pazocal{P}(g^{t-1})\,\mathbf{y}_{\text{ini}}+\pazocal{Q}(g^{t-1})\,\tilde{\mathbf{u}},
\end{equation}
where, with an abuse of notation, $g^t(\cdot)$ denotes the approximation of $g_{\text{\newacro}}(\cdot)$ with one iteration at time instant $t$. In this way, the \newacro\ predictor is approximated by a linear equality constraint that can be efficiently solved within a quadratic program. Thus, the proposed approach solves the following optimization problem at each time step
\begin{equation}
\begin{aligned}
    \underset{\mathbf{u},\mathbf{y}}{\text{minimize}} &\quad J_\text{ctr}(\mathbf{u},\mathbf{y})\\
    \text{subject to} &\quad g^t=\pazocal{P}(g^{t-1})\,\mathbf{y}_{\text{ini}}+\pazocal{Q}(g^{t-1})\,\tilde{\mathbf{u}},\\&\quad\mathbf{y}=Y_f\, g^t,\mathbf{u} \in \pazocal{U}, \mathbf{y} \in \pazocal{Y}.
\end{aligned}
\label{eqn:mlepc}
\end{equation}

The parameter estimation part (\ref{eqn:iddeepc}) in regularized DeePC has the same form as the maximum likelihood estimator (\ref{eqn:optsqp}) with $p=2$, which leads to the predictor in (\ref{eqn:mlepc}). However, our proposed method %avoids the aforementioned shortcomings with the regularized DeePC algorithm: 
isolates the parameter estimation part from the control performance objective. %Such approaches are known as indirect data-driven control \cite{dorfler2021bridging}.
More importantly, the problem of hyperparameter tuning is avoided by deriving the coefficients statistically, which requires only information about the noise levels of the offline data $\sigma^2$ and the online measurements $\sigma_p^2$.

\subsection{Numerical Results}

In this subsection, we compare the control performance of three receding horizon predictive control algorithms: 1) \textit{Sub-PC}: subspace predictive control (\ref{eqn:subpc}), 2) \textit{DeePC}: regularized DeePC (\ref{eqn:deepc}), and 3) \textit{\newacro-PC}: predictive control with the \newmodel\ (\ref{eqn:mlepc}). In \textit{DeePC}, the algorithm is tested on a nine-point logarithmic grid of $\lambda_g$ between 10 and 1000. In this example, the control performance is found to not be sensitive to the value of $\lambda_y$, so a fixed value of $\lambda_y=1000$ is used. In \textit{\newacro-PC}, the noise level $\sigma^2$ is estimated using (\ref{eqn:sigest}). Assuming the same sensor for offline and online measurements, we select $\sigma^2=\sigma_p^2$. To benchmark the performance, we also consider the ideal MPC algorithm (denoted by \textit{MPC}), where both the true state-space model and the noise-free state measurements are available. The result of this benchmark algorithm is thus deterministic and gives the best possible control performance with receding horizon predictive control.

In this example, we consider the LTI system (\ref{eqn:sys1}). Unless otherwise specified, the following parameters are used in the simulation: 
$
    N=200,\,L_0=4,\,L'=11,\,\sigma^2=\sigma_p^2=1,\,Q=R=1.
$
No input and output constraints are enforced, i.e., $\pazocal{U}=\mathbb{R}^{L' n_u}$ and $\pazocal{Y}=\mathbb{R}^{L' n_y}$. A square-wave reference trajectory labeled \textit{Ref} in Figure~\ref{fig:4}(a) is to be tracked. The offline data are generated with unit i.i.d. Gaussian input signals. For each case, 100 Monte Carlo simulations are conducted. In each run, 60 time steps are simulated. The control performance is assessed by the true stage cost over all time steps, i.e.,
\begin{equation}
    J=\sum_{k=0}^{N_c-1}\left(\norm{y_k^0-r_k}_Q^2+\norm{u_k}_R^2\right),
\end{equation}
where $N_c=60$ and $y_k^0$ is the true output at time $k$. 

When comparing the closed-loop performance, the best choices of $\lambda_g$ in \textit{DeePC} are selected with an oracle for each run as plotted in Figure~\ref{fig:7} (green) for different noise levels. It can be seen that, even for the same control task, the optimal value of this hyperparameter is not only sensitive to the noise level but also to the specific realization of the noise. This makes the tuning process difficult in practice. The optimal value of $\lambda_g$ is used in all the following simulations. On the other hand, the noise level estimator (\ref{eqn:sigest}) used in \textit{\newacro-PC} is very effective in estimating $\sigma^2$ as demonstrated in Figure~\ref{fig:7} (yellow).

\begin{figure}[htb]
\centering
\includegraphics[width=\columnwidth]{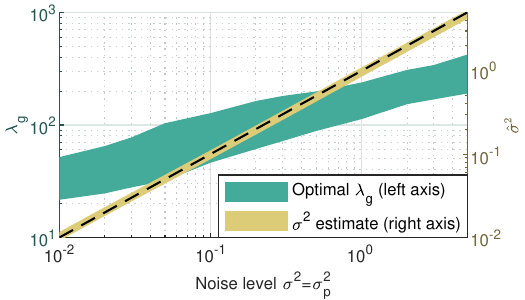}
\caption{Hyperparameter tuning in \textit{DeePC} ($\lambda_g$) and \textit{\newacro-PC} ($\sigma^2$) for different noise levels. Colored area: values within one standard deviation. The dashed line shows the true noise level.}
\label{fig:7}
\end{figure}

The optimization problems are all formulated as quadratic programming problems and solved by MOSEK. The computation time for all three algorithms is similar. The effect of the proposed data compression scheme in Section~\ref{sec:precon} is illustrated in Figure~\ref{fig:8} with the example of \textit{\newacro-PC}. By applying the preconditioning, the online computational complexity no longer depends on the data size $N$.

\begin{figure}[htb]
\centering
\includegraphics[width=\columnwidth]{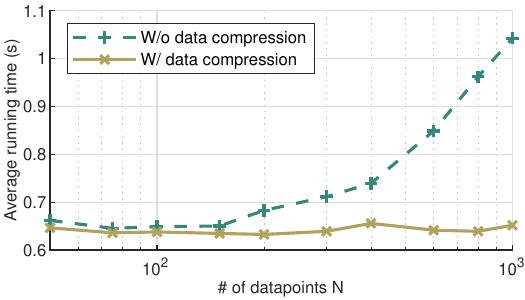}
\caption{Average computation time of \textit{\newacro-PC} with and without the data compression scheme.}
\label{fig:8}
\end{figure}

The closed-loop input-output trajectories of different control algorithms are plotted in Figure~\ref{fig:4}. The closed-loop trajectories are characterized by the range within one standard deviation of the Monte-Carlo simulation. The \textit{\newacro-PC} algorithm obtains the closest match to the benchmark trajectory \textit{MPC}. \textit{Sub-PC} applies more aggressive control inputs which results in much larger input costs, whereas the control strategy in \textit{DeePC} is more conservative which results in larger tracking errors. \textit{\newacro-PC} also has the smallest variance of input trajectories against different noise realizations. The boxplot of the control performance measure $J$ is shown in Figure~\ref{fig:5}, which confirms that \textit{\newacro-PC} performs better than \textit{Sub-PC} and \textit{DeePC} in this control task, even when the optimal tuning of $\lambda_g$, which is not realistic in practice, is used.

\begin{figure*}[htb]
\centering
\includegraphics[width=7.16in]{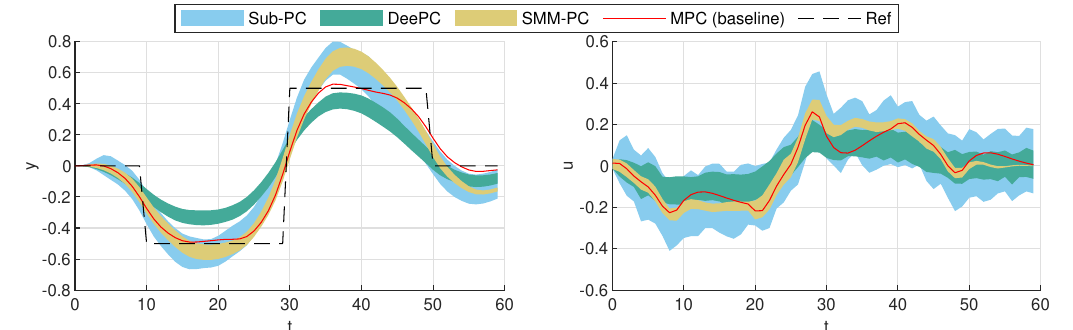}
\footnotesize (a) Output trajectory \hspace{23em} (b) Input trajectory
\caption{Comparison of closed-loop input-output trajectories with different control algorithms. Colored area: trajectories within one standard deviation.}
\label{fig:4}
\end{figure*}

\begin{figure}[htb]
\centering
\includegraphics[width=\columnwidth]{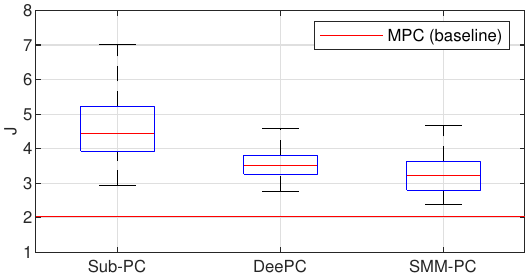}
\caption{Comparison of the control performance in terms of total stage costs $J$ with different control algorithms with 100 simulations ($\sigma^2=\sigma_p^2=1, N=200$).}
\label{fig:5}
\end{figure}

The effects of different offline data sizes $N$ and noise levels $\sigma^2,\sigma_p^2$ are investigated in Figure~\ref{fig:6}. As shown in Figure~\ref{fig:6}(a), the control performance of \textit{\newacro-PC} is not sensitive to the number of datapoints and performs uniformly better among the three algorithms. In fact, good performance is already obtained at only $N=75$. \textit{DeePC} does not perform very well with small data sizes but gets a similar performance to \textit{\newacro-PC} for large $N$. \textit{Sub-PC} cannot achieve a satisfying result even with a large data size because, as discussed in Section~\ref{sec:comp}, the subspace predictor is problematic to deal with online measurement noise $\sigma_p^2$, which cannot be averaged out by a large $N$. Figure~\ref{fig:6}(b) shows that all three algorithms perform similarly at low noise levels as they are all stochastic variants of the noise-free algorithm (\ref{eqn:deepc0}). \textit{\newacro-PC} obtains slightly worse results under low noise levels ($\sigma^2=\sigma_p^2 < 0.05$) compared to the optimally tuned \textit{DeePC} with an oracle, but the performance improvement of \textit{\newacro-PC} is significant for higher noise levels.

\begin{figure}[htb]
\centering
\includegraphics[width=\columnwidth]{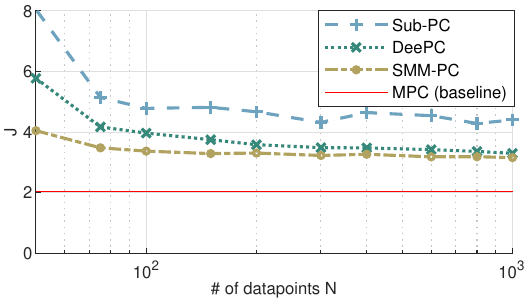}
\footnotesize (a) Performance as a function of the number of datapoints
\includegraphics[width=\columnwidth]{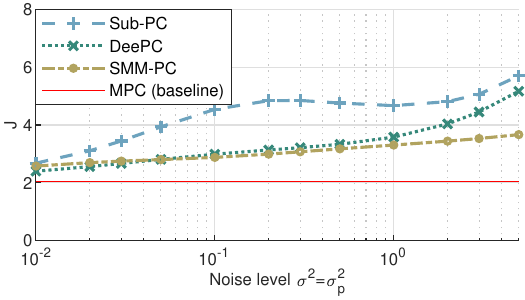}
\footnotesize (b) Performance as a function of the noise level
\caption{The effect of different offline data sizes and noise levels on the control performance.}
\label{fig:6}
\end{figure}

\section{Conclusions}
\label{sec:7}

In this work, we propose a novel statistical framework to estimate data-driven models from large noise-corrupted datasets. This is formulated as a maximum likelihood estimation problem. The problem is solved efficiently by approximating it as a sequential quadratic program with data compression and data-driven noise level estimation. This framework extends the current works on data-driven methods to noisy data by providing an optimal solution to the underdetermined implicit model structure and establishing the \newmodel.

With the \newmodel, two approaches in system identification and receding horizon control are developed. They obtain an impulse response estimate with less restrictive assumptions and an effective tuning-free data-driven receding horizon control algorithm respectively. The results from these two approaches demonstrate that the proposed framework can improve the state-of-the-art methods in both data-driven simulation and control in the presence of noisy data.\\

% Acknowledgement

\appendices

\section{\quad\ \ Derivation of $\Sigma_y$}
\label{sec:app1}
Let $\zeta_i\in \mathbb{R}^L$, $i=1,\dots,LM$ be the $i$-th column of $\left(g^\mathsf{T}\otimes \mathbb{I}\right)$, $\pazocal{S}=\left\{(i,j)\left|\left(\text{vec}(Y)\right)_i=\left(\text{vec}(Y)\right)_j\right.\right\}$, and $\Sigma_{y1}=\left(g^\mathsf{T}\otimes \mathbb{I}\right)\Sigma_{yd}\left(g\otimes \mathbb{I}\right)$. According to (\ref{eqn:ry}), we have
\begin{equation}
\Sigma_{y1}=\sigma^2\sum_{(i,j)\in\pazocal{S}}\zeta_i\zeta_j^\mathsf{T}.
\end{equation}
Let the $i$-th and the $j$-th entries of $\text{vec}(Y)$ correspond to the $(q,r)$-th and the $(s,t)$-th entries of $Y$ respectively, i.e., $i=(r-1)L+q$, $j=(t-1)L+s$. From the Hankel structure, the pair $(i,j)\in \pazocal{S}$ iff $q+r=s+t$. According to the structure of $\left(g^\mathsf{T}\otimes \mathbb{I}\right)$, we have $\zeta_i=g_r\mathbf{e}_q$, $\zeta_j=g_t\mathbf{e}_s$, where $\mathbf{e}_q\in \mathbb{R}^L$ is the unit vector with $q$-th non-zero entry, and similarly for $\mathbf{e}_s$. Thus,
\begin{equation}
\Sigma_{y1}=\sigma^2\sum_{q+r=s+t}g_r g_t\mathbf{e}_q\mathbf{e}_s^\mathsf{T}.
\end{equation}
So the $(q,s)$-th entry of $\Sigma_{y1}$ is given by
\begin{equation}
\left(\Sigma_{y1}\right)_{q,s}=\sigma^2\sum_{q+r=s+t}g_r g_t,
\end{equation}
which directly leads to (\ref{eqn:py}).

\section{\quad\ \ Proof of Proposition~\ref{prop:2}}
\label{sec:app2}
Let $K_\lambda=F^{-1}-F^{-1}U^\mathsf{T}(U F^{-1}U^\mathsf{T})^{-1}UF^{-1}$ and $g_\lambda=K_\lambda Y_p^\mathsf{T}\mathbf{y}_{\text{ini}}+\pazocal{Q}\tilde{\mathbf{u}}$. From the structure of (\ref{eqn:optsqp}), when $\lambda\rightarrow 0$, $g_\lambda$ converges to $g_{\text{pinv}}$. When $\lambda=L'\sigma_p^2/\norm{g_\lambda}_2^2$, $g_\lambda=g_{\text{\newacro}}$. Then we have
$
\text{cov}(g_\lambda)=\sigma_p^2(K_\lambda Y_p^\mathsf{T})(K_\lambda Y_p^\mathsf{T})^\mathsf{T}
$.
The derivative of $\text{tr}(\text{cov}(g_\lambda))$ with respect to $\lambda$ is calculated as follows.
\begin{equation}
    \begin{aligned}
    \dfrac{\partial\,\text{tr}(\text{cov}(g_\lambda))}{\partial \left(F^{-1}\right)_{i,j}}&=\text{tr}\left[\left(
    \dfrac{\partial\,\text{tr}(\text{cov}(g_\lambda))}{\partial K_\lambda}\right)^\mathsf{T}\dfrac{\partial K_\lambda}{\partial \left(F^{-1}\right)_{i,j}}\right]\\
    &=2\sigma_p^2\,\text{tr}\left[\left(Y_p^\mathsf{T}Y_p K_\lambda\right)^\mathsf{T}K_\lambda F\Delta(i,j) FK_\lambda\right],
    \end{aligned}
\end{equation}
where the $(i,j)$-th element of $\Delta(i,j)\in\mathbb{R}^{M\times M}$ is 1 and the other elements are 0. Then,
\begin{equation}
    \begin{aligned}
    \dfrac{\partial\,\text{tr}(\text{cov}(g_\lambda))}{\partial \lambda}&=\text{tr}\left[\left(
    \dfrac{\partial\,\text{tr}(\text{cov}(g_\lambda))}{\partial F^{-1}}\right)^\mathsf{T}\dfrac{\partial F^{-1}}{\partial \lambda}\right]\\
    &=-2\sigma_p^2\,\text{tr}\left[\left(FK_\lambda\left(Y_p^\mathsf{T}Y_p K_\lambda\right)^\mathsf{T}K_\lambda F\right)^\mathsf{T}F^{-2}\right]\\
    &=-2\sigma_p^2\,\text{tr}\left(K_\lambda Y_p^\mathsf{T}Y_p K_\lambda K_\lambda\right),
    \end{aligned}
\end{equation}
According to the Schur complement, since
\begin{equation}
    \begin{bmatrix}
    F^{-1}&F^{-1}U^\mathsf{T}\\UF^{-1}&UF^{-1}U^\mathsf{T}
    \end{bmatrix}=
    \begin{bmatrix}
    \mathbb{I}\\U
    \end{bmatrix}
    F^{-1}
    \begin{bmatrix}
    \mathbb{I}&U^T
    \end{bmatrix}\succ 0,
\end{equation}
we have $K_\lambda\succ 0$. Together with $K_\lambda Y_p^\mathsf{T}Y_p K_\lambda\succ 0$, we have
$\partial\,\text{tr}(\text{cov}(g_\lambda))/\partial \lambda < 0$
for all $\lambda$. This directly leads to Proposition~\ref{prop:2}.

\bibliographystyle{IEEEtran}
\bibliography{IEEEabrv,refs}

\end{document}